\documentclass[11pt,epsf]{article}
\usepackage{amsmath}
\usepackage{amsfonts}
\usepackage{amssymb}
\usepackage{graphicx}
\usepackage{color}
\usepackage{comment}
\usepackage{amsthm}
\usepackage{geometry}
\usepackage{hyperref}

\newtheorem{theorem}{Theorem}

\newtheorem{proposition}{Proposition}

\newtheorem{example}{Example}
\newtheorem{definition}{Definition}

\newcommand {\dfn} {\stackrel{\Delta} {=}}

\newcommand {\reals} {{\rm I\!R}}

\newcommand {\bone} {\mbox{\boldmath $1$}}
\newcommand{\hK}{{\hat{K}}}
\newcommand{\hP}{{\hat{P}}}
\newcommand{\hH}{{\hat{H}}}
\newcommand{\hX}{{\hat{X}}}
\newcommand{\hx}{{\hat{x}}}

\newcommand {\be} {\mbox{\boldmath $e$}}

\newcommand {\bw} {\mbox{\boldmath $w$}}
\newcommand {\bx} {\mbox{\boldmath $x$}}
\newcommand {\hbx} {\hat{\bx}}
\newcommand {\bsigma} {\mbox{\boldmath $\sigma$}}
\newcommand {\by} {\mbox{\boldmath $y$}}
\newcommand {\bz} {\mbox{\boldmath $z$}}

\newcommand {\bE} {\mbox{\boldmath $E$}}

\newcommand{\calA}{{\cal A}}
\newcommand{\calB}{{\cal B}}

\newcommand{\calI}{{\cal I}}

\newcommand{\calS}{{\cal S}}

\newcommand{\calW}{{\cal W}}
\newcommand{\calX}{{\cal X}}
\newcommand{\calY}{{\cal Y}}
\newcommand{\calZ}{{\cal Z}}
\allowdisplaybreaks

\begin{document}
\thispagestyle{empty}
\title{Generalized Forms of the Kraft Inequality for Finite-State Encoders}

\author{Neri Merhav}

\date{}
\maketitle

\begin{center}
The Andrew \& Erna Viterbi Faculty of Electrical Engineering\\
Technion - Israel Institute of Technology \\
Technion City, Haifa 3200003, ISRAEL \\
E--mail: {\tt merhav@ee.technion.ac.il}\\
\end{center}
\vspace{1.5\baselineskip}
\setlength{\baselineskip}{1.5\baselineskip}

\begin{abstract}
We derive a few extended versions of the Kraft inequality for
information lossless finite-state encoders. The main basic contribution is in
defining a notion of a Kraft matrix and in establishing the fact that a
necessary condition for information losslessness of a finite-state encoder is
that none of the eigenvalues of this matrix have modulus larger than unity, or
equivalently, the generalized Kraft inequality asserts that the 
spectral radius of the Kraft matrix cannot exceed one. For the important
special case where the FS encoder is irreducible,
we derive several equivalent forms of this inequality, which are based on
well known formulas for spectral radius. It also turns out that in the
irreducible case, Kraft sums are bounded by a constant, independent of the
block length, and thus cannot grow even in any subexponential rate.
Finally, two extensions are outlined - one concerns the case of side
information available to both encoder and decoder, and the other is for lossy
compression.
\end{abstract}

\section{Introduction}

Kraft's inequality plays a pivotal role in information theory.
It provides a complete and elegant characterization of the feasibility of
variable--length uniquely decodable (UD) codes
by imposing a simple constraint on codeword lengths.
In 1949, Kraft \cite{Kraft49} introduced this inequality for prefix codes,
establishing a condition on codeword lengths necessary for prefix decodability.
Seven years later, McMillan \cite{McMillan56} generalized this to UD codes,
leading to the Kraft-McMillan inequality, which is widely used in information
theory, first and foremost, to furnish a necessary and sufficient condition for
the existence of a UD code with a given code-length function,
and thereby also to prove the converse to the lossless source
coding theorem, asserting that no UD source code can yield a coding rate below
the entropy rate of the source.
Once this necessary and sufficient condition is satisfied,
there exists, not only a general UD code, but also more
specifically, a prefix code with that length function.
Beyond its immediate operational meaning, Kraft's inequality
underlies many fundamental principles in lossless compression, such as the
equivalence between lossless source coding and probability assignment.
In general, its importance stems from the fact that it connects combinatorial properties of codes
with analytical bounds in a precise and tractable manner.
Classical treatments can be found in standard texts such as \cite{Gallager68}
and \cite{CT06}.

When memory is introduced into the encoder, however,
the classical Kraft inequality (CKI) no longer applies directly.
Finite--state encoders constitute a natural and widely studied model for compression with memory, arising
in universal source coding, individual--sequence coding, and finite--state prediction.
In this setting, the encoder's output depends, not only on the current source
symbol, but also on an internal state
that evolves over time in a manner that depends on past inputs. As a result, the set of admissible
codeword length assignments is no longer characterized by a single scalar inequality,
and the extension of Kraft's condition becomes substantially more subtle.

Early progress in this direction was made by Ziv and Lempel \cite{ZL78}, who
derived a generalized Kraft inequality (GKI)
for information--lossless (IL) finite--state (FS) encoders by considering blocks over
large super-alphabets, see Lemma 2 in \cite{ZL78}. When reading Ziv and Lempel's article, the reader is
under the impression is that their GKI was
established merely an
auxiliary result needed on their way of proving that the finite-state compressibility
of a sequence is lower bounded by its asymptotic empirical entropy.
Their focus was not on the Kraft inequality on its own right. Consequently, their formulation
of Kraft's inequality  suffers from two main limitations: (i)
it does not reduce exactly to the CKI when the encoder
has merely one state, and (ii)
it is based on super-alphabet extensions to long blocks rather than being
formulated in a single-letter manner, or at the level in which the encoder is
defined in the first place.
More precisely, while the inequality remains valid even for short block
lengths, it yields tight results only asymptotically for long blocks.

Subsequent work has explored various aspects of finite--state coding,
including irreducibility, asymptotic equipartition properties, and connections to entropy and prediction.
Nevertheless, a direct, state--level generalization of Kraft's inequality that mirrors
the simplicity and sharpness of the classical result has remained elusive.

In this paper, we present several new forms of GKI's
for IL FS encoders. Our approach associates with every given IL FS encoder a nonnegative matrix,
termed the {\em Kraft matrix}, whose entries are determined by the encoder's single--symbol output lengths
and state transitions. We show that information losslessness imposes a spectral--radius constraint
on this matrix, which serves as a natural analogue of Kraft's inequality.
Unlike Ziv and Lempel's GKI mentioned above,
this inequality, as well as its several equivalent forms presented herein,
reduces exactly to the CKI in the single--state case
and avoids the use of super-alphabet extensions.

We then further refine the analysis for irreducible FS encoders, where Perron--Frobenius theory yields stronger,
uniform bounds on matrix powers. These results lead to transparent lower bounds on achievable compression rates
for both stochastic sources and individual sequences. In addition, we extend the framework to settings with
side information (SI) available at both the encoder and decoder, where the relevant constraint is expressed in
terms of the joint spectral radius (JSR) of a finite set of Kraft matrices
\cite{RS60}. This extension clarifies the structural limitations
imposed by SI and highlights the role of common sub-invariant vectors.
Finally, another extension is associated with lossy source coding in the
spirit of those of \cite{Campbell73}, \cite{me95}, and \cite{me26}.

Overall, the proposed framework provides a unified and exact characterization of
feasibility conditions for FS encoders, sharpening existing results and offering
new tools for the analysis of compression and prediction under finite--memory constraints.

The outline of the remaining part of this article is as follows.
In Section \ref{nsb}, we establish notation conventions, define the setting,
and provide some background on the GKI of Ziv and
Lempel. In Section \ref{bgki}, we present our basic GKI asserting that the
spectral radius of the Kraft matrix must not exceed unity for an IL FS
encoder. Stronger and more explicit statements are then provided for
irreducible encoders in Section \ref{irred}. In Section \ref{converses}, we
apply the GKI of Section \ref{irred} to obtain converse bounds on compression
and prediction of irreducible machines, both in the probabilistic setting and
for individual sequences. Finally, in Section \ref{si}, we extend the GKI to
the case of availability of SI, and in Section \ref{lossy}, we extend it to
the lossy case.

\section{Notation, Setting and Background}
\label{nsb}

Throughout this paper, scalar random
variables (RV's) will be denoted by capital
letters, their sample values will be denoted by
the respective lower case letters, and their alphabets will be denoted
by the respective calligraphic letters.
A similar convention will apply to
random vectors and their sample values,
which will be denoted with same symbols superscripted by the dimension.
Thus, for example, $X^n$ ($n$ -- positive integer)
will denote a random $n$-vector $(X_1,\ldots,X_n)$,
and $x^n=(x_1,\ldots,x_n)$ is a specific vector value in $\calX^n$,
the $n$--th Cartesian power of $\calX$, which is the alphabet of each
component of $x^n$. For two positive integers, $i$ and $j$, where $i\le j$,
$x_i^j$ and $X_i^j$
will designate segments $(x_i,\ldots,x_j)$
and $(X_i,\ldots,X_j)$, respectively,
where for $i=1$, the subscript will be omitted (as above).
For $i > j$, $x_i^j$ (or $X_i^j$) will be understood as the null string.
An infinite sequence $(x_1,x_2,\ldots)$ will be denoted by $\bx$.
Logarithms and exponents, throughout this paper, will be understood to be taken to the base 2
unless specified otherwise. The indicator of an event $\calA$ will be denoted
by $\calI\{\calA\}$, i.e., $\calI\{\calA\}=1$ if $\calA$ occurs and
$\calI\{\calA\}=0$ if not.

Following the the finite-state encoding model of \cite{ZL78}, 
an FS encoder is
defined by the quintuple,
$E=(\calX,\calY,\calZ,f,g)$,
whose five ingredients are defined as follows:
$\calX$ is the finite alphabet of each symbol of the source sequence to
be compressed. The cardinality of $\calX$ will be denoted by $\alpha$.
$\calY$ is a finite collection of binary variable-length strings, which is
allowed to consist of empty string, denoted `null' (whose length is zero);
$\calZ$ is a finite set of $s$ states of the encoder;
$f: \calZ \times \calX \to \calY$ is the output function, and
$g: \calZ \times \calX \to \calZ$ is the next-state function.

Given an infinite source sequence to be compressed,
$\bx = (x_1, x_2, \ldots)$, with $x_i \in
\calX$, the FS encoder $E$ produces an
infinite output sequence, $\by = (y_1, y_2, \ldots)$ with $y_i \in \calY$,
forming the compressed bit-stream, while passing through a sequence of
states $\bz = (z_1, z_2, \ldots)$ with $z_i \in \calZ$, $i=1,2,\ldots$.
The encoder is governed by the recursive equations:
\begin{eqnarray}
y_i &=& f(z_i, x_i), \label{yi} \\
z_{i+1}&=&g(z_i,x_i), \label{nextstate}
\end{eqnarray}
for $i = 1, 2, \ldots$, with a fixed initial state $z_1\in\calZ$.
If at any step $y_i = \mbox{null}$, this is referred to as idling as no output is
generated, but only the state evolves in response to the input. At each time
instant $i$, the encoder emits $L(y_i)=L[f(z_i,x_i)]$ bits, and it is
understood that $L(\mbox{null})=0$.

An encoder with $s$ states, henceforth called an $s$-state
encoder, is one for which $|\calZ| = s$. 
For the sake of simplicity, we adopt a few notation conventions from \cite{ZL78}:
Given a segment of input symbols $x_i^j$, where $i$ and $j$ are positive
integers with $i \le j$, and an initial state
$z_i$, we use $f(z_i, x_i^j)$ to denote the corresponding output
segment $y_i^j$ produced by $E$.
Similarly, $g(z_i, x_i^j)$ will denote the final state $z_{j+1}$ after processing the
inputs $x_i^j$, beginning from state $z_i$. Thus, in response to an input
$x^n$, the encoder produces a compressed bit string of length
$L(y^n)=L[f(z_1,x^n)]=\sum_{i=1}^nL[f(z_i,x_i)]=\sum_{i=1}^nL(y_i)$ bits.
An FS encoder $E$ is called {\em information lossless} (IL) if,
given any initial state $z_i \in \calZ$, any positive integer $n$, and any
input string, $x_i^{i+n}$, the triplet
$(z_i,f(z_i,x_i^{i+n}),g(z_i,x_i^{i+n}))$
uniquely determines the corresponding input string $x_i^{i+n}$.

For example, a fixed-to-variable block encoder of length $k$ can be viewed as
an IL FS encoder with
$s=\sum_{j=0}^{k-1}\alpha^j=\frac{\alpha^k-1}{\alpha-1}$ states. The initial state
designates the beginning of each block. At each time instant, the state of the
encoder is simply the contents
of the part of the current input block received so far. In general, as long as the input block has not been
completed, the encoder idles 
and upon receiving the last input symbol of the block, the encoder produces the
compressed codeword for that block
and it returns to its initial state, ready to receive the next input block. In some cases, if there is enough
structure, the encoder does not necessarily have to
idle until the very end of the block. For instance, if at a certain time before
the end of the block, the contents of the beginning of the block read so far
already determines the beginning of the compressed representation, the encoder
can start to output these compressed bits before the end of the block.
As an instance of such a block code of length $k=2$, see Example \ref{xmp1} below.

In Lemma 2 of \cite{ZL78}, 
Ziv and Lempel presented a GKI for IL FS encoders. It
asserts that
for every IL encoder with $s$ states and every positive integer $\ell$,
\begin{equation}
\label{ZL78Kraft}
\sum_{x^\ell\in\calX^\ell} 2^{-\min_{z\in{\cal Z}}L[f(z,x^\ell)]}\le
s^2\left[1+\log\left(1+\frac{\alpha^\ell}{s^2}\right)\right].
\end{equation}
Ziv and Lempel's GKI was a perfect tool for their purpose of
proving that the compression ratio achieved by an IL FS encoder cannot be smaller
than the asymptotic empirical entropy rate (defined in \cite{ZL78}) for any infinite source sequence
$\bx$. However, when examined for finitely long sequences, and from the
perspective of serving as a necessary condition for information losslessness,
this inequality suffers from two main weaknesses.
\begin{enumerate}
\item It does not
exactly recover the CKI for the special case, $s=1$, as
in that case, the
r.h.s.\ becomes $1+\log(1+\alpha^\ell)>1$. Moreover, even if $\ell=1$, the
right-hand side (r.h.s.),
which is $1+\log_2(1+\alpha)$, is even larger than $2$ for every $\alpha\ge 2$.
On a related note, a close inspection of the proof of Lemma 2 in \cite{ZL78} reveals that the
inequality in eq.\ (\ref{ZL78Kraft}) is 
actually a strong inequality ($<$), in other words, this inequality is always loose.
\item It is significant only upon
an extension from single symbols into the super-alphabet of $\ell$-strings for large $\ell$, unlike the ordinary
Kraft inequality, which is asserted in the same level that code is defined.
For example, the CKI for a code that is defined in the level of
single symbols of $\calX$ is asserted in that level, i.e.,
$\sum_{x\in{\cal X}}2^{-L[f(x)]}\le 1$.
\end{enumerate}
Our objective in this work is first and foremost, to establish another GKI
for IL FS encoders that is free of the above mentioned drawbacks.
In other words, for the case $s=1$, it would recover the traditional Kraft
inequality exactly, and it will be posed in the single-letter level without
recourse to alphabet extensions. The latter property will enable one to verify 
relatively easily that this inequality holds in a given situation.

Our first proposed GKI serves as the basis for our subsequent derivations.
Having derived it, we then confine attention to the subclass of
irreducible IL FS encoders, namely, FS encoders for which every state can be
reached from every state in a finite number of steps. For this important subclass of
encoders, we provide several alternative formulations of the GKI and
provide a stronger upper bound to the growth rate of the Kraft sum as function
of the block length. Again, all these forms are smooth extensions of the
CKI in the sense that in the special case $s=1$, they
degenerate to the CKI. Finally, we consider extensions in two
directions (one at a time): the first is the case where SI is
available to both encoder and decoder, and the second is the case of lossy
compression.

\section{The Basic Generalized Kraft Inequality}
\label{bgki}

For a given IL FS encoder $E$ with $s$ states,
let us define an $s\times s$ {\em Kraft matrix} $K$, whose $(z,z')$ entry is given by
\begin{equation}
\label{Kdef}
K_{zz'}=\sum_{\{x:~g(z,x)=z'\}}2^{-L[f(z,x)]},~~~~~(z,z')\in\calZ^2,
\end{equation}
where the summation over an empty set is understood as zero.
Since $K$ is a non-negative matrix, then according to Theorem 8.3.1 in
\cite{HJ85}, the spectral radius of $K$, $\rho(K)$, is
an eigenvalue of $K$.
Our first form of a GKI is the following.\\

\begin{theorem}
\label{thm1}
For every IL FS encoder,
\begin{equation}
\label{basicgki}
\rho(K)\le 1.
\end{equation}
\end{theorem}

As can be seen, this GKI has the two desired properties we mentioned above:
\begin{enumerate}
\item The case $s=1$
obviously recovers the CKI, since in this case, $K$
degenerates to a scalar, which is nothing but the Kraft sum,
$\sum_{x\in{\cal X}}2^{-L[f(x)]}$, and then eq.\ (\ref{basicgki}) asserts that
$\sum_{x\in{\cal X}}2^{-L[f(x)]}\le 1$, as desired.
\item The matrix $K$ is defined in terms of the functions $f$ and $g$ only.
These functions are defined in
the level of the single symbols and states. 
\end{enumerate}

The first property sets the stage of establishing the condition $\rho(K)\le 1$
as a {\em necessary condition for information losslessness} of a given FS
encoder, in analogy to the fact that ordinary Kraft inequality is a necessary (and
sufficient) condition for the existence of unique decodability in the case
$s=1$. Since there is no involvement of summations over super-alphabets of long vectors,
this condition is relatively easy to check, similarly as the CKI, which is
a necessary condition for the unique decodability (UD) property of ordinary lossless source
codes.

\begin{proof}
Let $L_{\max}\dfn\max_{z,x}L[f(z,x)]$.
For every positive integer $\ell$, the $(z,z')$ entry of the $\ell$-th order
power, $K^\ell$, is given by
\begin{eqnarray}
\label{lineargrowth}
[K^\ell]_{zz'}&=&\sum_{z_2\in{\cal Z}}\sum_{z_3\in{\cal
Z}}\cdot\cdot\cdot\sum_{z_\ell\in{\cal Z}}\prod_{i=1}^{\ell}
\left(\sum_{\{x_i:~g(z_i,x_i)=z_{i+1}\}}2^{-L[f(z_i,x_i)]}\right)\nonumber\\
&=&\sum_{\{x^\ell:~g(z,x^\ell)=z'\}}2^{-L[f(z,x^\ell)]}\nonumber\\
&=&\sum_{l=0}^{\ell\cdot L_{\max}}
2^{-l}\cdot|\{x^\ell:~L[f(z,x^\ell)]=l,~g(z,x^\ell)=z'\}|\nonumber\\
&\le&\sum_{l=0}^{\ell\cdot L_{\max}}
2^{-l}\cdot 2^l\nonumber\\
&=&1+\ell\cdot L_{\max},
\end{eqnarray}
where in the first line, $z_1=z$ and $z_{\ell+1}=z'$, and the inequality is
due to the postulated IL property (as $z$ and $z'$ are fixed).
Alternatively, we can also bound $[K^\ell]_{zz'}$ by $1+\log(1+\alpha^\ell)$
using the same considerations as in the proof of Lemma 2 in \cite{ZL78}, except that the factor $s^2$ is
missing since $z$ and $z'$ are fixed. Which bound is tighter depends on
$L_{\max}$. In any case, both
expressions are essentially linear in $\ell$.
Continuing with the first bound, it follows that
\begin{equation} 
\label{eq7}
\sum_{z'\in\calS}[K^\ell]_{zz'}=\sum_{x^\ell\in\calX^\ell}
2^{-L[f(z,x^\ell)]}\le s(1+\ell\cdot L_{\max}).
\end{equation}
Let $\be_z$ be a column vector of dimension $s$ whose entries are all zero except
the entry corresponding to state $z$, which is 1, and let $\bone$ denote the all-one column
vector of dimension $s$. Then, eq.\ (\ref{eq7}) can be rewritten as
\begin{equation}
\label{lineargrowth1}
\be_z^{\top}K^\ell \bone\le s(1+\ell\cdot L_{\max}).
\end{equation}
To prove that $\rho(K)\le 1$,
assume conversely, that $\lambda\dfn\rho(K)>1$.
Since $K$ has non-negative entries, the Perron--Frobenius theorem
(see again Theorem 8.3.1 in \cite{HJ85}) guarantees
that the right eigenvector $v$ corresponding to $\lambda$ has non-negative
components and at least one strictly positive component.
Since $\bone=(1,\dots,1)^T$ has strictly positive components,
there exists a constant $\delta>0$ such that
$\bone\ge\delta v$ component-wise.
Multiplying by $K^\ell$ from the left and using the non-negativity of $K$,
we obtain
\begin{equation}
K^\ell \bone \ge \delta K^\ell v
= \delta \lambda^\ell v .
\end{equation}
Taking the $z$-th component yields
\begin{equation}
\be_z^{\top} K^\ell \bone \ge \delta \lambda^\ell v_z .
\end{equation}
For any index $z$ with $v_z>0$, the r.h.s.\ grows exponentially
in $\ell$ since $\lambda>1$, but
this contradicts eq.\ (\ref{lineargrowth1})
which establishes an upper bound that grows only linearly in $\ell$.
Therefore the postulate $\rho(K)>1$ cannot hold true, and we conclude that
$\rho(K)\le 1$, which completes the proof.
\end{proof}

Since $\rho(K)\le 1$, it is clear that for every natural $\ell$,
$\rho(K^\ell)=[\rho(K)]^\ell\le\rho(K)\le 1$. In other words,
the spectral radius of
\begin{equation}
K^\ell=\left\{\sum_{\{x^\ell:~g(z,x^\ell)=z'\}}2^{-L[f(z,x^\ell)]}\right\}_{z,z'\in{\cal
Z}} 
\end{equation}
is also never larger than unity, which is an extension of our GKI to
super-alphabets, which is again, a smooth extension that degenerates to the CKI
for $s=1$.\\

\begin{example}
\label{xmp1}
Consider a binary source sequence and a block code of length 2, which maps the source strings 00, 01,
10, and 11, into
0, 10, 110, and 111, respectively. This code can be implemented by a FS
encoder with $s=3$ states, labeled `S', `O', and `I', using the
following functions, $f$ and $g$ (see also Fig.\ \ref{fig1}):
\begin{eqnarray}
g(\mbox{S},0)&=&\mbox{O},\nonumber\\
g(\mbox{S},1)&=&\mbox{I},\nonumber\\
g(\mbox{O},0)&=&g(\mbox{O},1)=g(\mbox{I},0)=g(\mbox{I},1)=\mbox{S},\nonumber
\end{eqnarray}
and
\begin{eqnarray}
f(\mbox{S},0)&=&\mbox{null},\nonumber\\
f(\mbox{S},1)&=&11,\nonumber\\
f(\mbox{O},0)&=&0,\nonumber\\
f(\mbox{O},1)&=&10,\nonumber\\
f(\mbox{I},0)&=&0,\nonumber\\
f(\mbox{I},1)&=&1.\nonumber
\end{eqnarray}
State `S' designates the start of a block. State `O' remembers
that the first input of the block was `0' and state `I' remembers that
the first input was `1'. Upon moving to state `I', the encoder can already
output `11', because the entire codeword will be either `110' or `111' if the
first source symbol is `1', so the first two coded bits are `11' in either
case. After state `I', the encoder can complete the codeword according to the second
input in the block. After state `O',
outputs are generated only upon receiving the second symbol.
After both states `O' and `I', the encoder must
return to state `S' in order to start the next block.
The corresponding Kraft matrix (with row and column indexing in the order
of (S,O,I)) is given by:
\begin{equation}
K=\left(\begin{array}{ccc}
0 & 2^{-0} & 2^{-2}\\
2^{-1}+2^{-2} & 0 & 0\\
2^{-1}+2^{-1} & 0 & 0\end{array}\right)=
\left(\begin{array}{ccc}
0 & 1 & 0.25\\
0.75 & 0 & 0\\
1 & 0 & 0\end{array}\right)
\end{equation}
whose eigenvalues are $1$, $0$, and $-1$, and so the spectral radius is
$\rho(K)=1$.
As can be seen,
the sums of the second and third rows do not exceed unity, so when the
initial state is either `O' or `I', the Kraft sum does not exceed 1. 
On the other hand, the Kraft sum corresponding to the first row (pertaining to
`S') exceeds unity. This demonstrates an important observation: The model of a
general IL FS encoder is broader and more general than a model of a FS encoder
for which given every state, the encoder implements a
certain prefix (or UD) code for the variety of incoming symbols.
For $\ell=100$, we find that
$$K^{100}=\left(\begin{array}{ccc}
1 & 0 & 0\\
0 & 0.75 & 0.1875\\
0 & 1 & 0.25\end{array}\right)$$
with eigenvalues are 0, 1, and 1.
Here, the sums of the first and the second rows do not exceed unity, so when the
initial state is either `S' or `O', the Kraft sum does not exceed 1.
On the other hand, the Kraft sum corresponding to the third row exceeds unity,
and so, the above comment with regard to $K$ applies here too.
This concludes the Example 1.
\end{example}

\begin{figure}[h]
\hspace*{3cm}
\begin{picture}(0,0)%
\includegraphics{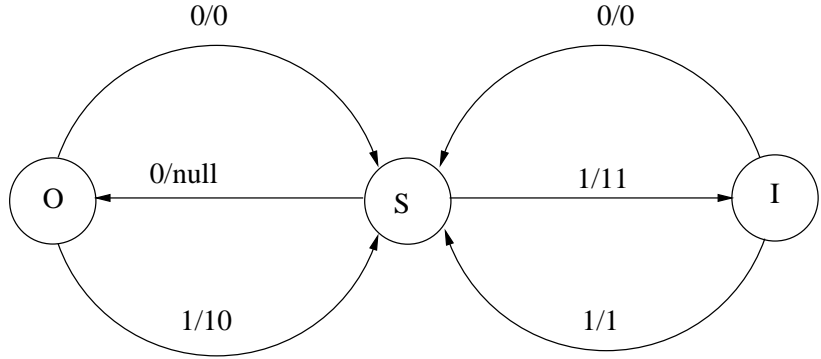}%
\end{picture}%
\setlength{\unitlength}{4144sp}%
\begingroup\makeatletter\ifx\SetFigFont\undefined%
\gdef\SetFigFont#1#2#3#4#5{%
  \reset@font\fontsize{#1}{#2pt}%
  \fontfamily{#3}\fontseries{#4}\fontshape{#5}%
  \selectfont}%
\fi\endgroup%
\begin{picture}(4853,2138)(187,-1526)
\end{picture}%
\caption{State transition diagram of the encoder in Example 1.
The various state transitions are labeled in a form $x/y$, where $x$ denotes
the input and $y=f(z,x)$ denotes the output.}
\label{fig1}
\end{figure}

Earlier we said that $\rho(K)\le 1$ is a necessary condition for a given code with
next-state function $g$ and code-lengths $\{L[f(z,x)]\}$ to be IL.
One might naturally wonder whether it is also a
sufficient condition.
This question is open in general, but we have two comments related to this
issue.

The first is that the answer is obviously affirmative for the subclass of IL encoders,
which satisfies the CKI for each and every state, i.e.,
$\sum_{z'\in\calZ}K_{zz'}=\sum_{x\in\calX} 2^{-L[f(z,x)]}\le 1$: Simply
construct a separate prefix code with length function $\{L[f(z,x)],~x\in\calX\}$
for each $z\in\calZ$. However, in general, an IL code does not necessarily
satisfy the ordinary Kraft inequality for each $z$. 
Indeed, in Example \ref{xmp1}, the sum of the first row of $K$ is larger than 1.

The second comment is that we can give an affirmative
answer in the level of longer blocks. Let $z_1\in\calZ$ be an arbitrary initial
state and consider the lengths, $l(z_1,x^n)=\sum_{i=1}^nL[f(z_i,x_i)]$. Then, as
we have seen in (\ref{lineargrowth}):
\begin{equation}
\sum_{x^n}2^{-L[f(z_1,x^n)]}\le s(1+nL_{\max}),
\end{equation}
where the factor of $s$ stems from taking the sum of $K_{zz'}$ over
$z'\in\calZ$.
Equivalently,
\begin{equation}
\sum_{x^n}2^{-[l[f(z_1,x^n)]+\log[s(1+nL_{\max})]}\le 1,
\end{equation}
and so, there exists a prefix code with lengths
$l'(x^n)=l(z_1,x^n)+\log[s(1+nL_{\max})]+\log s$, which are relatively only slightly longer
than those of the original code. Here, the additional $\log s$ term is a
header that notifies $z_1$. 

\section{Irreducible FS Encoders}
\label{irred}

IL FS encoders for which the next-state function $g$ allows transition from every
state to every state within a finite number of steps, are henceforth referred to as {\em irreducible FS
encoders}. Equivalently, defining the $s\times s$ adjacency matrix $A$ such
that $A_{zz'}=1$ whenever $\exists x\in\calX$ such that $g(z,x)=z'$ and 
$A_{zz'}=0$ otherwise, then an IL FS encoder is irreducible iff the matrix $A$
is irreducible. Likewise, an IL FS encoder is irreducible iff the matrix $K$
is irreducible. For an irreducible FS encoder, the shortest path from every
state $z$ to every state $z'$ lasts no longer than $s-1$ steps, because any longer path
must visit a certain state $z''$ at least twice, meaning that this path contains
a loop starting and ending at $z''$, which can be eliminated.
Clearly, the encoder of Example 1 is irreducible.

Intuitively, it makes sense to use irreducible encoders,
because for reducible ones, once the machine leaves a certain subset of
transient states, it can never return, and so, effectively, a reducible encoders uses
eventually a smaller number of states after finite time. Specifically, given a
reducible machine and an infinite individual sequence $x_1,x_2,\ldots$, 
suppose the machine starts at a transient state. Then, there
are two possibilities: either the machine quits the subset of transient
states after finite time, or it stays in that subset forever. In the former
case, the transient states are in use for finite time only and then never used
again. In the latter case, the recurrent states are never used. In either case,
asymptotically, only a subset of the available states are used, and so,
effectively, the number of states actually used is smaller than $s$.
Let $\calZ_\infty$
denote the set of states visited infinitely many times along the sequence. 
This set is necessarily closed and induces a strongly connected subgraph. 
Consequently, the asymptotic behavior of the encoder along the given sequence 
is governed entirely by its restriction to $\calZ_\infty$,
which constitutes an irreducible finite-state encoder with strictly fewer than 
$s$ states. Therefore, reducible encoders cannot offer asymptotic advantages over irreducible ones, even for individual sequences.

Assume next that the next-state function $g$ induces an irreducible matrix
$K^\ell$, where $\ell$ be an arbitrary positive integer. Since $K^\ell$ is
non-negative and irreducible,
the Collatz-Wielandt formulas \cite{Collatz42}, \cite{Wielandt50} for the spectral radius of $K^\ell$ hold true.
These are given by
\begin{equation}
\rho(K^\ell)=\max_{w\in\calW_+}\min_{\{z:~w_z>0\}}\frac{[K^\ell w]_z}{w_z}
=\min_{w\in\calW_+}\max_{\{z:~w_z>0\}}\frac{[K^\ell w]_z}{w_z}.
\end{equation}
where $w$ is an $s$-dimensional column vector and $\calW_+$ is the set of
all such vectors with 
non-negative components not all of which are zero.
These lead to the two following GKI's:
\begin{equation}
\forall~w\in\calW_+~\exists~z~\mbox{such that}~w_z>0~\mbox{and}~\sum_{z'\in{\cal S}}
w_{z'}\cdot\sum_{\{x^\ell:~g(z,x^\ell)=z'\}}2^{-L[f(z,x^\ell)]}\le w_z,
\end{equation}
and
\begin{equation}
\exists~w\in\calW_+~\forall~z~\mbox{such that}~w_z>0\mbox{:}~\sum_{z'\in{\cal S}}
w_{z'}\cdot\sum_{\{x^\ell:~g(z,x^\ell)=z'\}}2^{-L[f(z,x^\ell)]}\le w_z,
\end{equation}
The first formulation can be simplified at the price of a possible
loss of tightness, by selecting $w$ to be the all-one vector and thereby
bounding $\rho(K^\ell)$ from below. This results in
the conclusion that an IL FS encoder always satisfies yet another GKI:
\begin{equation}
\label{lessthanone}
\exists z\in{\cal S}~
\sum_{x^\ell\in{\cal X}^\ell}2^{-L[f(z,x^\ell)]}\le 1.
\end{equation}
In words, for every given irreducible FS encoder, $(f,g)$,
and for every natural $\ell$, there is at least one initial state,
$z\in{\cal Z}$, for which the Kraft sum is less than unity, but again, not all
states must satisfy this condition (as we saw in Example 1, the Kraft sum
exceeds unity when the initial state is `S').
All these are also smooth extensions of the
CKI in the sense that for $s=1$ we are back to the CKI.

But there is an even stronger GKI that applies to irreducible encoders.
It asserts that in the irreducible case, $K^n$ does not even grow
linearly as in (\ref{lineargrowth}), but is rather bounded by a constant, independent of $n$. For $s=1$,
this constant is 1, again in agreement with the CKI.

\begin{theorem}
\label{thm2}
Let $K$ be an irreducible Kraft matrix. Then, for all $z,z'\in\calZ$
and for every natural $n$,
\begin{equation}
(K^n)_{zz'}=\sum_{\{x^n:~g(z,x^n)=z'\}}2^{-L[f(z,x^n)]}\le 2^{(s-1)L_{\max}}.
\end{equation}
Consequently, for every $z\in\calS$,
\begin{equation}
\sum_{x^n\in\calX^n}2^{-L[f(z,x^n)]}\le s\cdot 2^{(s-1)L_{\max}},
\end{equation}
and
\begin{equation}
\sum_{z\in\calS}\sum_{x^n\in\calX^n}2^{-L[f(z,x^n)]}\le s^2\cdot 2^{(s-1)L_{\max}}.
\end{equation}
\end{theorem}

\begin{proof}
It is sufficient to prove the first inequality, as the two other ones will follow
trivially by a summation over $z'\in\calZ$ and then also over $z\in\calZ$,
respectively. Since $K$ is non-negative and irreducible, the Perron-Frobenius theorem
applies. This theorem asserts that the spectral radius, $\rho(K)$, 
is positive and simple, with left and right
eigenvectors, $u$ and $v$, respectively, that have only strictly positive components.
In Theorem \ref{thm1} we have already proved that $\rho(K)\le 1$. Assume first that $\rho(K)=1$.
Then, $u^TK^n=u^T$,
or, equivalently,
\begin{equation}
\sum_{z\in\calS}u_z (K^n)_{zz'}=u_{z'}~~~~\forall~z'\in\calZ.
\end{equation}
Since all terms are non-negative, the left-hand side is lower bound by
$u_z(K^n)_{zz'}$ for any $z\in\calZ$. This implies for every $z,z'\in\calZ$
\begin{equation}
(K^n)_{zz'}\le\frac{u_{z'}}{u_z}\le\frac{\max_{z\in\calZ}u_z}{\min_{z\in\calZ}u_z}.
\end{equation}
Let $z_\star\in\calZ$ and $z^\star\in\calZ$ be achievers of $\min_{z\in\calZ}u_z$ and
$\max_{z\in\calZ}u_z$, respectively. Then, for every $z,z'\in\calZ$,
\begin{equation}
(K^n)_{zz'}\le\frac{u_{z^\star}}{u_{z_\star}}.
\end{equation}
Since $K$ is irreducible, the exists a path of length $\ell\le s-1$ from
$z^\star$ to $z_\star$, say, $z^\star\to z_1\to \cdot\cdot\cdot\to
z_{\ell-1}\to z_\star$ such that
\begin{equation}
(K^\ell)_{z^\star z_\star}\ge K_{z^\star z_1}\cdot K_{z_1z_2}\cdot\cdot\cdot K_{z_{\ell-1}z_\star}>0. 
\end{equation}
Since all positive
entries of $K$ are at least as large as $2^{-L_{\max}}$, this product is
at least as large as $2^{-\ell L_{\max}}\ge 2^{-(s-1)L_{\max}}$.
It follows then that
\begin{equation}
(K^\ell)_{z^\star z_\star}\ge 2^{-(s-1)L_{\max}}.
\end{equation}
Now,
\begin{equation}
u_{z_\star}=\sum_{z\in\calS}u_z(K^\ell)_{zz_\star}\ge
u_{z^\star}(K^\ell)_{z^\star z_\star}\ge
u_{z^\star}2^{-(s-1)L_{\max}},
\end{equation}
which implies that
\begin{equation}
2^{(s-1)L_{\max}}\ge \frac{u_{z^\star}}{u_{z_\star}}\ge (K^n)_{zz'},
\end{equation}
for every $z,z'\in\calS$. This
completes the proof for the case $\rho(K)=1$. The case $\rho(K)< 1$ is
obtained from the case $\rho(K)=1$ by simply defining $\hK=K/\rho(K)$ and using the fact that all
non-negative entries of $\hK$ are lower bounded by $2^{-L_{\max}}/\rho(K)$.
Since $\hK$ is also irreducible and since $\rho(\hK)=1$, we now have
\begin{equation}
(\hK^n)_{zz'}\le [\rho(K)2^{L_{\max}}]^{s-1}.
\end{equation}
But $\hK^n=K^n/[\rho(K)]^n$, and so,
\begin{equation}
(K^n)_{zz'}\le[\rho(K)]^{n+s-1}\cdot 2^{(s-1)L_{\max}}< 2^{(s-1)L_{\max}}.
\end{equation}
This completes the proof of Theorem \ref{thm2}.
\end{proof}

\section{Converse Bounds Derived from the GKI}
\label{converses}

In this section, we demonstrate how the GKI of Section \ref{irred} can be used
to obtain lower bounds on the performance of irreducible machines in
compression and in prediction problems. For compression, both probabilistic
sources and individual sequences are considered. For prediction, only the
individual sequence version is presented, but the probabilistic counterpart
can also be derived straightforwardly using the same ideas.

Let $\{P(z,x^\ell),~z\in\calZ,~x^\ell\in\calX^\ell\}$ be a joint probability
distribution of random variables $Z$ and $X^\ell$. Then,
\begin{eqnarray}
s^2\cdot
2^{(s-1)L_{\max}}&\ge&\sum_{z\in\calZ}\sum_{x^\ell\in\calX^\ell}2^{-L[f(z,x^\ell)]}\nonumber\\
&=&\sum_{z\in\calZ}\sum_{x^\ell\in\calX^\ell}P(z,x^\ell)\cdot
2^{-L[f(z,x^\ell)]-\log P(z,x^\ell)}\nonumber\\
&\ge&\exp_2\left\{-\sum_{z\in\calZ}\sum_{x^\ell\in\calX^\ell}P(z,x^\ell)
L[f(z,x^\ell)]+H(Z,X^\ell)\right\}\nonumber\\
&=&\exp_2\left[-\bE\{
L[f(Z,X^\ell)]\}+H(Z,X^\ell)\right],
\end{eqnarray}
where the inequality follows from Jensen's inequality and the convexity of the
exponential function.
By taking logarithms of both sides, rearranging terms, and normalizing
by $\ell$, we get
\begin{eqnarray}
R&=&\frac{\bE\{L[f(Z,X^\ell)]\}}{\ell}\nonumber\\
&\ge&\frac{H(Z,X^\ell)}{\ell}-\frac{\log_2\left(s^2\cdot
2^{(s-1)L_{\max}}\right)}{\ell}\nonumber\\
&\ge&\frac{H(X^\ell)}{\ell}-\frac{2\log_2s+
(s-1)L_{\max}}{\ell},
\end{eqnarray}
and if the source $P$ is stationary, $H(X^\ell)/\ell$ can be further lower
bounded by $H(X_\ell|X^{\ell-1})$, to obtain
\begin{equation}
R\ge H(X_\ell|X^{\ell-1})-
\frac{2\log_2s+
(s-1)L_{\max}}{\ell}.
\end{equation}
Since this bound applies to every positive integer $\ell$, we may maximize the
lower bound over $\ell$, and obtain
\begin{equation}
R\ge \sup_{\ell\ge 1}\left\{H(X_\ell|X^{\ell-1})-
\frac{2\log_2s+
(s-1)L_{\max}}{\ell}\right\}.
\end{equation}
We see that thanks to Theorem \ref{thm2}, the vanishing term subtracted from the
entropy decays at the rate of $1/\ell$ as opposed to the $(\log\ell)/\ell$
rate that stems from Lemma 2 of \cite{ZL78} as well as from the more general inequality of
$1+\ell\cdot L_{\max}$, that is obtained when reducible machines are allowed.

In the context of individual sequences, we can arrive at an analogous lower
bound, provided that we define a shift-invariant empirical distribution.
Specifically, let $x^n$ be a given individual
sequence, let $\ell$ be a positive integer smaller than $n$, and let $z_1$ be a given initial state of the
encoder. We assume that $x^n$ cyclic with respect to (w.r.t.) $g$ in the
sense that $g(z_n,x_n)=z_1$. If this is not the case, consider an
extension of $x^n$ by concatenating a suffix $x_{n+1}^{n+m}$ such that the
extended sequence would be cyclic w.r.t.\ $g$. Since $g$ is assumed irreducible, this
is always possible and the length $m$ of the extension need not be larger than
$s-1$. To avoid cumbersome notation, we redefine $x^n$ to be the sequence
after the cyclic extension (if needed), and we shall keep in mind that 
this cyclic extension adds no more than $m\cdot L_{\max}\le (s-1)L_{\max}$ bits
to the compressed description, or equivalently, $(s-1)L_{\max}/n$ to the
compression ratio, and so, this extra rate should be subtracted back upon
returning to the original sequence before the cyclic extension.
For every $w^\ell\in\calX^\ell$ and $z\in\calS$, let
\begin{equation}
\delta(z_i,x_i^{((i-1)\oplus(\ell-1))+1};z,w^\ell)=\left\{\begin{array}{ll}
1 & z_i=z~\mbox{and}~x_i^{((i-1)\oplus(\ell-1))+1}=w^\ell\\
0 & \mbox{elsewhere}\end{array}\right.
\end{equation}
where $\oplus$ denotes modulo-$n$ addition. Next,
define the empirical distribution
\begin{equation}
\hP(z,w^\ell)=\frac{1}{n}\sum_{i=1}^{n}\delta(z_i,x_i^{((i-1)\oplus(\ell-1))+1};z,w^\ell),
\end{equation}
Now,
\begin{eqnarray}
\frac{1}{n}\sum_{i=1}^{n}L[f(z_i,x_i)]&=&
\frac{1}{n\ell}\sum_{i=1}^{n}\ell\cdot L[f(z_i,x_i)]\nonumber\\
&=&\frac{1}{n\ell}\sum_{i=1}^{n}\sum_{j=0}^{\ell-1}L[f(z_i,x_{((i-1)\oplus j)+1})]\nonumber\\
&=&\frac{1}{n\ell}\sum_{i=1}^{n}L[f(z_i,x_i^{((i-1)\oplus (\ell-1))+1})]\nonumber\\
&=&\frac{1}{n\ell}\sum_{i=1}^{n}\sum_{z\in\calS}\sum_{w^\ell\in\calX^\ell}
\delta(z_i,x_i^{((i-1)\oplus(\ell-1))+1};z,w^\ell)
L[f(z,w^\ell)]\nonumber\\
&=&\frac{1}{n\ell}
\sum_{z\in\calS}\sum_{w^\ell\in\calX^\ell}
\sum_{i=1}^{n}
\delta(z_i,x_i^{((i-1)\oplus(\ell-1))+1};z,w^\ell)
L[f(z,w^\ell)]\nonumber\\
&=&\frac{1}{\ell}\sum_{z\in\calS}\sum_{w^\ell\in\calX^\ell}\hP(z,w^\ell)L[f(z,w^\ell)]\nonumber\\
&\ge&\hH(X_{\ell}|X^{\ell-1})-
\frac{2\log_2s+
(s-1)L_{\max}}{\ell},
\end{eqnarray}
where $\hH(X_\ell|X^{\ell-1})$ is the empirical conditional entropy derived
from the shift-invariant distribution $\hP$. Using the fact that this is true for every
natural $\ell < n$ and returning to the original sequence before the cyclic
extension, we find that
\begin{equation}
\frac{1}{n}\sum_{i=1}^{n}L[f(z_i,x_i)]\ge\max_{1\le \ell<
n}\left\{\hH(X_\ell|X^{\ell-1})-\frac{2\log
s+(s-1)L_{\max}}{\ell}\right\}-\frac{(s-1)L_{\max}}{n}.
\end{equation}
Furthermore, invoking Ziv's inequality (see eq.\ (13.125) in \cite{CT06}), this can be further lower bounded in
terms of the LZ complexity. Specifically, according to eq.\ (13.125) in
\cite{CT06},
for every Markov source, $Q_{\ell-1}$, of order $\ell-1$ and every
$x^n\in\calX^n$,
\begin{equation}
c(x^n)\log c(x^{n})\le -\log
Q_{\ell-1}(x^{n}|x_{-(\ell-2)}^{0})+\epsilon_\ell(n),
\end{equation}
where $c(x^{n})$ is the maximum number of distinct phrases whose concatenation
forms $x^n$, and
where $\epsilon_\ell(n)$ tends to zero at the rate of $O(\log(\log n)/\log n)$
for every fixed $\ell$. By minimizing the r.h.s.\ w.r.t.\ $Q_{\ell-1}$, we get
\begin{equation}
c(x^{n})\log c(x^{n})\le n\hH(X_\ell|X_0^{\ell-1})
+n\cdot\epsilon_\ell(n),
\end{equation}
and so,
\begin{equation}
\frac{1}{n}\sum_{i=1}^{n}L[f(z_i,x_i)]\ge
\frac{c(x^{n})\log c(x^{n})}{n}-\min_\ell\left[\epsilon_\ell(n)+\frac{2\log
s+(s-1)L_{\max}}{\ell}\right]-\frac{(s-1)L_{\max}}{n}.
\end{equation}
The minimizing $\ell$ can be found to be proportional to $\sqrt{n}$, but the
dominant term of $\epsilon_\ell(n)$ remains of the order 
of $\frac{\log(\log n)}{\log n}$.

We next derive a lower bound to the prediction error of any FS predictor that
is based on an irreducible FS machine.
Consider a finite-state (FS) predictor with $q$ states, defined by the following
recursion, for $i=1,2,\ldots$
\begin{eqnarray}
\label{recursion}
\hx_{i+1}&=&u(x_i,\sigma_i),\nonumber\\
\sigma_{i+1}&=&v(x_i,\sigma_i),
\end{eqnarray}
where $\bsigma=(\sigma_1,\sigma_2,\ldots)$, $\sigma_i\in\Sigma$, $i=1,2,\ldots$, is a
corresponding infinite state sequence, whose alphabet, $\Sigma$, is a finite set of
states of cardinality $q$, and $\hbx=(\hx_1,\hx_2,\ldots)$, $\hx_i\in\calX$,
$i=1,2,\ldots$, is the resulting
predictor output sequence. Without loss of generality, the initial state,
$\sigma_1$, and the initial prediction,
$\hx_1$, are assumed fixed members, $\sigma_\star\in\Sigma$ and $\hx_\star\in\calX$, respectively. Here,
$u:\calX\times\Sigma\to\calX$ is the predictor output function and
$v:\calX\times\Sigma\to\Sigma$ is the next-state function.

It is assumed that $\calX$ is a group with well-defined addition and
subtraction operations. For example, if $\calX=\{0,1,\ldots,\alpha-1\}$ then
it is natural to equip $\calX$ with addition and subtraction modulo
$\alpha$. Let $\rho:\calX\to\reals^+$ denote a given loss function. Then,
the performance of a predictor across the time range, $1\le t\le n$ is measured in terms 
of the time-average,
\begin{equation}
\frac{1}{n}\sum_{i=1}^{n}\rho(x_i-\hx_i).
\end{equation}
Given an arbitrary irreducible FS predictor $(u,v)$ as defined above, consider the
auxiliary conditional probability distribution,
\begin{equation}
Q_\theta(x_{i+1}|x_i,\sigma_i)=
\frac{e^{-\rho(x_{i+1}-u(x_i,\sigma_i))/\theta}}{Z(\theta)},~~\theta\ge 0,
\end{equation}
where 
\begin{equation}
Z(\theta)=\sum_{x\in\calX}e^{-\rho(x)/\theta}.
\end{equation}
Define also the function
\begin{equation}
\Delta(R)=\sup_{\theta\ge 0} \theta\cdot[R-\log Z(\theta)],~~~R\ge 0.
\end{equation}
Now, define 
\begin{equation}
Q_\theta(x^{n})=\prod_{i=1}^{n}Q_\theta(x_i|x_{i-1},\sigma_{i-1})
\end{equation}
where $\sigma_0$ and $x_0$ are arbitrary members of $\Sigma$ and $\calX$,
respectively, such that $\sigma_1=v(x_{0},\sigma_{0})=s_\star$, and
$\sigma_2,\sigma_3,\ldots,\sigma_{n-1}$
are generated from $x_1,x_2,\ldots,x_{n-1}$ as in (\ref{recursion}).

Let $k$ divide $n$ and consider the lossless compression of $x_0^{n-1}$ in blocks of
length $k$, $\bx_j=x_{jk+1}^{jk+k}$, $j=0,1,\ldots,n/k-1$, by using the
Shannon code, whose length function for a vector $x^k$ is $\lceil-\log
Q_\theta(x^k)\rceil$. 
This is equivalent to predictive coding, where the
prediction error signal, $z_n=x_n-f(x_{n-1},s_{n-1})$ is compressed losslessly
under a model of a memoryless source with a marginal $Q_\theta(z)$ (see Fig.\
\ref{fig2} for illustration).
In this case, since the ceiling operation is carried
over $k$-blocks, and there are $n/k$ such $k$-blocks, the upper bound to
$L(x^n)$ becomes
\begin{equation}
L(x^n)=\sum_{i=0}^{n/k-1}\lceil-\log Q_\theta(x_{ik+1}^{ik+k})\rceil\le
\frac{1}{\theta}\cdot\sum_{i=1}^{n}\rho(x_i-u(x_{i-1},\sigma_{i-1}))+n\log
Z(\theta)+\frac{n}{k}.
\end{equation}
On the other hand, the corresponding encoder of Fig.\ \ref{fig2} can be viewed
as an encoder with $q\cdot M_k$ states, where $M_k=(\alpha^k-1)/(\alpha-1)$,
since this is the number of combinations of
a state of the $q$-state predictor and a state of the lossless block encoder,
whose number of states is $\sum_{j=0}^{k-1}\alpha^j
=M_k$. Thus,
\begin{equation}
\frac{L(x^{n})}{n}\ge\hH(X^\ell|X^{\ell-1})-\frac{2\log(qM_k)+(qM_k-1)L_{\max}}{\ell}-\frac{L_{\max}}{n}
\end{equation}
where it should be kept in mind that $L_{\max}$ is expected to grow linearly
with $k$. Thus, by comparing the upper bound and the lower bound to $L(x^n)$,
we have
\begin{eqnarray}
& &\frac{1}{n\theta}\cdot\sum_{i=1}^{n}\rho(x_i-u(x_{i-1},\sigma_{i-1}))+\log
Z(\theta)+\frac{1}{k}\nonumber\\
&\ge&\hH(X^\ell|X^{\ell-1})-\frac{2\log(qM_k)+(qM_k-1)L_{\max}}{\ell}-\frac{L_{\max}}{n}.
\end{eqnarray}
or, equivalently,
\begin{eqnarray}
& &\frac{1}{n}\sum_{i=1}^{n}\rho(x_i-u(x_{i-1},\sigma_{i-1}))\nonumber\\
&\ge&\theta\left[\hH(X^\ell|X^{\ell-1})-\frac{2\log(qM_k)+(qM_k-1)L_{\max}}{\ell}-\frac{L_{\max}}{n}-\frac{1}{k}-\log
Z(\theta)\right].
\end{eqnarray}
Maximizing the r.h.s\ over $\theta\ge 0$, we get
\begin{equation}
\frac{1}{n}\sum_{i=1}^{n}\rho(x_{i+1}-u(x_i,\sigma_i))\ge\Delta\left(\hH(X_\ell|X^{\ell-1})-
\frac{2\log(qM_k)+(qM_k-1)L_{\max}}{\ell}-\frac{L_{\max}}{n}-\frac{1}{k}\right).
\end{equation}
The bound is meaningful if $k\gg 1$ and $\ell\gg qM_k$, so that the two
subtracted terms in the argument of the function $\Delta(\cdot)$ are small
compared to the main term, $\hH(X_\ell|X^{\ell-1})$.
It is tight essentially for sequences of the form
$x_i=u(x_{i-1},\sigma_{i-1})+z_i$, $i=1,2,\ldots$, where $z^n=(z_1,\ldots,z_n)$ is typical to
an i.i.d.\ source and
where the marginal empirical distribution of each $z_i$ is close to
$e^{-\rho(z)/\theta}/Z(\theta)$ for some $\theta\ge 0$.

\begin{figure}[h]
\hspace*{0cm}
\begin{picture}(0,0)%
\includegraphics{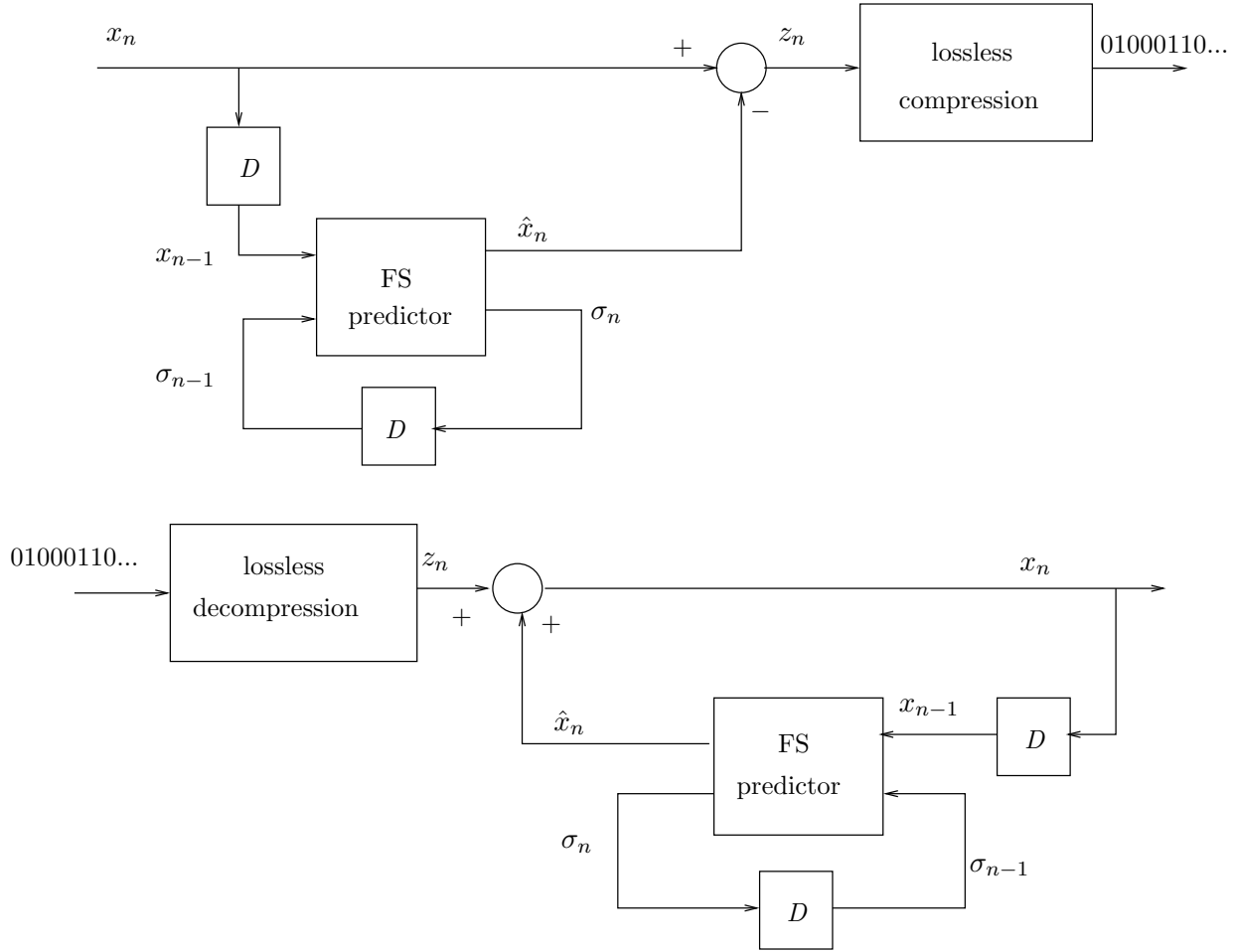}%
\end{picture}%
\setlength{\unitlength}{4144sp}%
\begingroup\makeatletter\ifx\SetFigFont\undefined%
\gdef\SetFigFont#1#2#3#4#5{%
  \reset@font\fontsize{#1}{#2pt}%
  \fontfamily{#3}\fontseries{#4}\fontshape{#5}%
  \selectfont}%
\fi\endgroup%
\begin{picture}(7468,5718)(481,-5137)
\put(1074,348){\makebox(0,0)[lb]{\smash{{\SetFigFont{11}{13.2}{\rmdefault}{\mddefault}{\itdefault}{\color[rgb]{0,0,0}$x_n$}%
}}}}
\put(3549,-833){\makebox(0,0)[lb]{\smash{{\SetFigFont{11}{13.2}{\rmdefault}{\mddefault}{\itdefault}{\color[rgb]{0,0,0}$\hat{x}_n$}%
}}}}
\put(4951,-118){\makebox(0,0)[lb]{\smash{{\SetFigFont{10}{12.0}{\rmdefault}{\mddefault}{\itdefault}{\color[rgb]{0,0,0}$-$}%
}}}}
\put(4484,266){\makebox(0,0)[lb]{\smash{{\SetFigFont{10}{12.0}{\rmdefault}{\mddefault}{\itdefault}{\color[rgb]{0,0,0}$+$}%
}}}}
\put(1871,-476){\makebox(0,0)[lb]{\smash{{\SetFigFont{10}{12.0}{\rmdefault}{\mddefault}{\itdefault}{\color[rgb]{0,0,0}D}%
}}}}
\put(2751,-2044){\makebox(0,0)[lb]{\smash{{\SetFigFont{10}{12.0}{\rmdefault}{\mddefault}{\itdefault}{\color[rgb]{0,0,0}D}%
}}}}
\put(3989,-1329){\makebox(0,0)[lb]{\smash{{\SetFigFont{11}{13.2}{\rmdefault}{\mddefault}{\itdefault}{\color[rgb]{0,0,0}$\sigma_n$}%
}}}}
\put(2531,-1357){\makebox(0,0)[lb]{\smash{{\SetFigFont{10}{12.0}{\rmdefault}{\mddefault}{\updefault}{\color[rgb]{0,0,0}predictor}%
}}}}
\put(2724,-1136){\makebox(0,0)[lb]{\smash{{\SetFigFont{10}{12.0}{\rmdefault}{\mddefault}{\updefault}{\color[rgb]{0,0,0}FS}%
}}}}
\put(1898,-2870){\makebox(0,0)[lb]{\smash{{\SetFigFont{10}{12.0}{\rmdefault}{\mddefault}{\updefault}{\color[rgb]{0,0,0}lossless}%
}}}}
\put(1596,-3117){\makebox(0,0)[lb]{\smash{{\SetFigFont{10}{12.0}{\rmdefault}{\mddefault}{\updefault}{\color[rgb]{0,0,0}decompression}%
}}}}
\put(496,-2815){\makebox(0,0)[lb]{\smash{{\SetFigFont{10}{12.0}{\rmdefault}{\mddefault}{\updefault}{\color[rgb]{0,0,0}01000110...}%
}}}}
\put(5117,348){\makebox(0,0)[lb]{\smash{{\SetFigFont{11}{13.2}{\rmdefault}{\mddefault}{\itdefault}{\color[rgb]{0,0,0}$z_n$}%
}}}}
\put(6574,-2842){\makebox(0,0)[lb]{\smash{{\SetFigFont{11}{13.2}{\rmdefault}{\mddefault}{\itdefault}{\color[rgb]{0,0,0}$x_n$}%
}}}}
\put(2971,-2815){\makebox(0,0)[lb]{\smash{{\SetFigFont{11}{13.2}{\rmdefault}{\mddefault}{\itdefault}{\color[rgb]{0,0,0}$z_n$}%
}}}}
\put(5117,-3942){\makebox(0,0)[lb]{\smash{{\SetFigFont{10}{12.0}{\rmdefault}{\mddefault}{\updefault}{\color[rgb]{0,0,0}FS}%
}}}}
\put(4869,-4189){\makebox(0,0)[lb]{\smash{{\SetFigFont{10}{12.0}{\rmdefault}{\mddefault}{\updefault}{\color[rgb]{0,0,0}predictor}%
}}}}
\put(6602,-3914){\makebox(0,0)[lb]{\smash{{\SetFigFont{10}{12.0}{\rmdefault}{\mddefault}{\itdefault}{\color[rgb]{0,0,0}D}%
}}}}
\put(3769,-3806){\makebox(0,0)[lb]{\smash{{\SetFigFont{11}{13.2}{\rmdefault}{\mddefault}{\itdefault}{\color[rgb]{0,0,0}$\hat{x}_n$}%
}}}}
\put(5171,-4960){\makebox(0,0)[lb]{\smash{{\SetFigFont{10}{12.0}{\rmdefault}{\mddefault}{\itdefault}{\color[rgb]{0,0,0}D}%
}}}}
\put(6272,-4657){\makebox(0,0)[lb]{\smash{{\SetFigFont{11}{13.2}{\rmdefault}{\mddefault}{\itdefault}{\color[rgb]{0,0,0}$\sigma_{n-1}$}%
}}}}
\put(6041,230){\makebox(0,0)[lb]{\smash{{\SetFigFont{10}{12.0}{\rmdefault}{\mddefault}{\updefault}{\color[rgb]{0,0,0}lossless}%
}}}}
\put(5847,-61){\makebox(0,0)[lb]{\smash{{\SetFigFont{10}{12.0}{\rmdefault}{\mddefault}{\updefault}{\color[rgb]{0,0,0}compression}%
}}}}
\put(1368,-1707){\makebox(0,0)[lb]{\smash{{\SetFigFont{11}{13.2}{\rmdefault}{\mddefault}{\itdefault}{\color[rgb]{0,0,0}$\sigma_{n-1}$}%
}}}}
\put(1368,-981){\makebox(0,0)[lb]{\smash{{\SetFigFont{11}{13.2}{\rmdefault}{\mddefault}{\itdefault}{\color[rgb]{0,0,0}$x_{n-1}$}%
}}}}
\put(3813,-4516){\makebox(0,0)[lb]{\smash{{\SetFigFont{11}{13.2}{\rmdefault}{\mddefault}{\itdefault}{\color[rgb]{0,0,0}$\sigma_n$}%
}}}}
\put(5847,-3693){\makebox(0,0)[lb]{\smash{{\SetFigFont{11}{13.2}{\rmdefault}{\mddefault}{\itdefault}{\color[rgb]{0,0,0}$x_{n-1}$}%
}}}}
\put(3151,-3166){\makebox(0,0)[lb]{\smash{{\SetFigFont{10}{12.0}{\rmdefault}{\mddefault}{\itdefault}{\color[rgb]{0,0,0}$+$}%
}}}}
\put(3691,-3211){\makebox(0,0)[lb]{\smash{{\SetFigFont{10}{12.0}{\rmdefault}{\mddefault}{\itdefault}{\color[rgb]{0,0,0}$+$}%
}}}}
\put(7058,278){\makebox(0,0)[lb]{\smash{{\SetFigFont{10}{12.0}{\rmdefault}{\mddefault}{\updefault}{\color[rgb]{0,0,0}01000110...}%
}}}}
\end{picture}%
\caption{Auxiliary predictive encoder and decoder. The upper block diagram
depicts the encoder that losslessly compresses the prediction error signal,
$z_n$, which is
the difference between the input signal, $x_n$, and its prediction,
$\hat{x}_n$ obtained using a FS predictor. The lower block diagram stands for
the corresponding decoder.}
\label{fig2}
\end{figure}

\section{GKI in the Presence of Side Information}
\label{si}

We now discuss briefly an extension of the GKI 
for IL FS encoders in the case where SI is available at both the encoder and the decoder. 
The resulting condition is expressed in terms of the {\em joint spectral
radius} (JSR) of a finite set of nonnegative matrices indexed by the various side--information symbols. 
We identify verifiable sufficient conditions for subexponential growth of Kraft sums and discuss the 
limitations inherent in the presence of SI.

Let $\mathcal{X}$ be the source alphabet as before and let $\mathcal{W}$
denote the finite alphabet of the SI sequence, $w_1,w_2,\ldots$, whose symbols
are synchronized with the corresponding source symbols. As before, let 
$\mathcal{Z}$ be the finite set of states with $|\mathcal{Z}|=s$.
An FS encoder with SI is specified by an output function
$f:\mathcal{Z}\times\mathcal{X}\times\mathcal{W}\to\calY$,
($\calY$ being defined as a subset of $\{0,1\}^*$, similarly as before) and a next--state function
$g:\mathcal{Z}\times\mathcal{X}\times\mathcal{W}\to\mathcal{Z}$.
Given an initial state, $z_1=z$, a source sequence, $\bx=(x_1,x_2,\ldots)$, 
and a SI sequence, $\bw=(w_1,w_2,\ldots)$, the encoder implements the
equations:
\begin{eqnarray}
y_i&=&f(z_i,x_i,w_i),\nonumber\\
z_{i+1}&=&g(z_i,x_i,w_i),
\end{eqnarray}
for $i=1,2,\ldots$,
and the total code-length produced by the encoder after $n$ steps is
\begin{equation}
L[f(z,x^n,w^n)]=
\sum_{i=1}^n L\bigl(f(z_i,x_i,w_i)\bigr).
\end{equation}

\begin{definition}
An FS encoder is said to be 
{\em information--lossless with side information} if for every $n$, the
quadruple $(z_1,y^n,w^n,z_{n+1})\in\calZ\times\calY^n\times\calW^n\times\calZ$
dictates $x^n\in\calX^n$.
\end{definition}

For each SI symbol, $w\in\mathcal{W}$, define the corresponding Kraft matrix
\begin{equation}
[K(w)]_{zz'}=
\sum_{\{x\in\mathcal{X}:~g(z,x,w)=z'\}}
2^{-L[f(z,x,w)]},
\qquad z,z'\in\mathcal{Z}.
\end{equation}
Each $K(w)$ is a nonnegative $s\times s$ matrix.
For a given SI sequence, $w^n$, define the product matrix
\begin{equation}
K(w^n)=K(w_1)\cdot K(w_2)\cdots K(w_n).
\end{equation}
Now, let $\mathcal{K}=\{K(w),~w\in\calW\}$. The growth rate of the Kraft products,
$K(w^n)$, over arbitrary SI sequences, $\{w^n\}$, is governed by the
JSR of $\mathcal{K}$, which is defined as follows.

\begin{definition}
The JSR of $\mathcal{K}$ is defined as
\begin{equation}
\rho_{\mathrm{JSR}}(\mathcal{K})=
\lim_{n\to\infty}
\max_{w^n\in\mathcal{W}^n}
\|K(w^n)\|^{1/n},
\end{equation}
where $\|\cdot\|$ is any matrix norm.
\end{definition}

It is a classical result that this limit exists and is independent of the chosen norm.
The GKI in the presence of SI can be formulated as
follows.

\begin{theorem}
For an IL FS encoder with SI, 
\begin{equation}
\rho_{\mathrm{JSR}}(\mathcal{K}) \le 1.
\end{equation}
\end{theorem}

\begin{proof}
Fix an arbitrary SI sequence, $w^n\in\mathcal{W}^n$ and states $z,z'\in\mathcal{Z}$. The $(z,z')$ entry of $K(w^n)$ is given by
\begin{equation}
\sum_{\{x^n:\ g(z,x^n,w^n)=z'\}}
2^{-L[f(z,x^n,w^n)]}.
\end{equation}
Since the encoder is IL for the fixed sequence $w^n$, the mapping
between $x^n$ and $(z,x^n,w^n,z')$
is injective over all paths from $z$ to $z'$. Grouping sequences according to their total code-length 
(similarly as before) and using a standard counting argument yields a linear 
upper bound (in $n$) on each matrix entry of $K(w^n)$, uniformly over $w^n$. 
Exponential growth of $\|K(w^n)\|$ is therefore impossible, 
and the JSR must satisfy $\rho_{\mathrm{JSR}}(\mathcal{K})\le1$.
\end{proof}

The following proposition can sometimes help.

\begin{proposition}
\label{prop1}
If there exists a vector $v\in\mathbb{R}^s$ with strictly positive components such that
$K(w) v \le v$ for every $w\in\calW$,
then for every SI sequence $w^n$,
$K(w^n)v \le v$,
and hence the family $\{K(w^n)\}$ is uniformly bounded.
\end{proposition}

\begin{proof}
The claim follows by induction on $n$. Since $v>0$, uniform boundedness 
of all products implies $\rho_{\mathrm{JSR}}(\mathcal{K})\le 1$.
\end{proof}

For example, if $v=\bone$ satisfies proposition \ref{prop1}, this means that the
Kraft sum is less than or equal to unity for every initial state and every SI
sequence. In such a case, one can simply design a separate prefix code for
every combination of initial state and SI sequence.

In contrast to the case without SI, bounding the spectral radius of each
individual Kraft matrix $K(w)$ is necessary but insufficient to control 
the growth rate of arbitrary products. In other words, even if $\rho[K(w)]\le 1$ for every
$w\in\calW$ individually, the JSR may exceed unity, and in
fact, may be arbitrarily large.
As an example, let $\epsilon$ be an arbitrarily small positive real and consider the matrices
$$A=\left(\begin{array}{cc}
\epsilon & \frac{1}{\epsilon}\\
0 & \epsilon\end{array}\right)$$
and $B=A^T$. While $\rho(A)=\rho(B)=\epsilon$, which is arbitrarily small, it turns out that $\rho(A\cdot
B)=\epsilon^2+\frac{1}{2\epsilon^2}+\sqrt{1+\frac{1}{4\epsilon^4}}\approx\frac{1}{\epsilon^2}$,
which is accordingly, arbitrarily large.
The JSR is therefore the correct quantity governing feasibility.

Exact computation of the JSR is undecidable in general, 
even for nonnegative rational matrices. Consequently, the above result should be interpreted as a 
structural constraint rather than a computational criterion. Nonetheless,
there is a plethora of upper and lower bounds to the JSR.
Also, as mentioned earlier, the existence of a common positive sub-invariant vector provides a meaningful 
and verifiable sufficient condition for subexponential growth.\\

\section{GKI for Lossy Compression}
\label{lossy}

For lossy compression, we adopt a simple encoder model, where each source
vector $x^\ell\in\calX^\ell$ is first mapped into a reproduction vector
$\hx^\ell=Q(x^\ell)\in\hat{\calX}^\ell$ within distortion $\ell D$ and then
$\hx^\ell$ is losslessly compressed by an IL FS
encoder with $s$ states exactly as before. The latter may work in the level of single letters or
in the level of $\ell$-blocks. Let us define
$\calB(\hx^\ell)=\{x^\ell\in\calX^\ell:~d(x^\ell,\hx^\ell)\le \ell D\}$
and let $B_\ell=\max_{\hbx}|\calB(\hbx)|$. Now,
\begin{eqnarray}
K_{zz'}&\dfn&\sum_{\{x^\ell:~g(z,Q(x^\ell))=z'\}} 2^{-L[f(z,Q(x^\ell))]}\nonumber\\
&=&\sum_{\{\hx^\ell:~g(z,\hx^\ell))=z'\}}\sum_{\{x^\ell:~Q(x^\ell)=\hx^\ell\}}
2^{-L[f(z,\hx^\ell)]}\nonumber\\
&\le&\sum_{\{\hx^\ell:~g(z,\hx^\ell))=z'\}}|\calB(\hx^\ell)|\cdot
2^{-L[f(z,\hx^\ell)]}\nonumber\\
&\le&B_\ell\cdot\sum_{\{\hx^\ell:~g(z,\hx^\ell))=z'\}}2^{-L[f(z,\hx^\ell)]}\nonumber\\
&\dfn&B_\ell\cdot \hat{K}_{zz'},
\end{eqnarray}
and so, $K\le B_\ell\cdot \hat{K}$ entry-wise.
Now, $\hat{K}$ has all the properties that we have proved for the lossless
case, it is just defined in the super-alphabet of $\ell$-blocks.
Since $\rho(\hat{K})\le 1$, we readily have:
\begin{equation}
\rho(K)=\rho(B_\ell\cdot\hat{K})=B_\ell\cdot\rho(\hat{K})\le B_\ell.
\end{equation}
For additive distortion measures, the quantity $B_\ell$ can be estimated using the method of types \cite{CK11}, or the
Chernoff bound, or saddle-point integration \cite{deBruijn81}, \cite{MW25}. It is is upper bounded by
$2^{\ell\Phi(D)}$, where
\begin{equation}
\Phi(D)=\max_{\{P_{X\hat{X}}:~d(X,\hX)\le D\}} H(X|\hX).
\end{equation}
Thus, the corresponding GKI reads
\begin{equation}
\rho(K)\le 2^{\ell\Phi(D)}.
\end{equation}


\end{document}